\documentclass[letterpaper, 10 pt, conference]{ieeeconf} 
\IEEEoverridecommandlockouts   
\usepackage{myStyle}
\providecommand{\tsl}{\ensuremath{T_{l}}\xspace}
\renewcommand{\tsl}{\ensuremath{T_{l}}\xspace}
\providecommand{\tsh}{\ensuremath{T_{h}}\xspace}
\renewcommand{\tsh}{\ensuremath{T_{h}}\xspace}

\providecommand{\fsl}{\ensuremath{f_{l}}\xspace}
\renewcommand{\fsl}{\ensuremath{f_{l}}\xspace}
\providecommand{\fsh}{\ensuremath{f_{h}}\xspace}
\renewcommand{\fsh}{\ensuremath{f_{h}}\xspace}

\providecommand{\fac}{\ensuremath{F}\xspace}
\renewcommand{\fac}{\ensuremath{F}\xspace}

\providecommand{\dt}{\ensuremath{n}\xspace}
\renewcommand{\dt}{\ensuremath{n}\xspace}

\DeclareDocumentCommand{\theoremRef}{m}{Theorem~\ref{#1}}
\DeclareDocumentCommand{\assRef}{m}{Assumption~\ref{#1}}
\DeclareDocumentCommand{\figRef}{m}{Fig.~\ref{#1}}
\DeclareDocumentCommand{\tabRef}{m}{Table~\ref{#1}}
\DeclareDocumentCommand{\secRef}{m}{Section~\ref{#1}}

\DeclareDocumentCommand{\exampleRef}{m}{Example~\ref{#1}}

\newcommand{\varn}{\operatorname{var}}

\usepackage{eso-pic}
\AddToShipoutPictureBG*{%
	\AtPageUpperLeft{%
		\setlength\unitlength{1in}%
		\hspace*{\dimexpr0.5\paperwidth\relax}
		\makebox(0,-0.5)[c]{\begin{tabular}{c c}
				Max van Haren \textit{et al}, Beyond Nyquist in Frequency Response Function Identification: Applied to Slow-Sampled Systems, \\
				To appear in {\em IEEE Control Systems Letters (L-CSS)}, uploaded to arXiv \today \\
		\end{tabular}}
}}

\AddToShipoutPictureBG*{%
	\AtPageUpperLeft{%
		\setlength\unitlength{1in}%
		\hspace*{\dimexpr0.5\paperwidth\relax}
		\makebox(0,-21)[c]{
			\footnotesize
			\begin{tabular}{c c}
				© 2023 IEEE. Personal use of this material is permitted. Permission from
				IEEE must be obtained for all other uses, in any \\
				current or future media, including reprinting/republishing
				this material for advertising or promotional purposes, creating new \\
				collective works, for resale or redistribution to servers or lists,
				or reuse of any copyrighted component of this work in other works.
\end{tabular}}}}

\begin{document}
	\bstctlcite{BSTControl}
	\title{Beyond Nyquist in Frequency Response Function Identification:\\ Applied to Slow-Sampled Systems}

\author{Max van Haren$^{1}$, Leonid Mirkin$^{2}$, Lennart Blanken$^{1,3}$ and Tom Oomen$^{1,4}$
	\thanks{This work is part of the research programme VIDI with project number 15698, which is (partly) financed by the Netherlands Organisation for Scientific Research (NWO). In addition, this research has received funding from the ECSEL Joint Undertaking under grant agreement 101007311 (IMOCO4.E). The Joint Undertaking receives support from the European Union Horizon 2020 research and innovation programme. It is also supported by the Israel Science Foundation (grant no.\,3177/21).}
	\thanks{$^{1}$Max van Haren, Lennart Blanken and Tom Oomen are with the Control Systems Technology Section, Department of Mechanical Engineering, Eindhoven University of Technology, Eindhoven, The Netherlands, e-mail: {\tt\small m.j.v.haren@tue.nl}.}%
	\thanks{$^{2}$Leonid Mirkin is with the Faculty of Mechanical Engineering, Technion---IIT, Haifa 3200003, Israel.}
	\thanks{$^{3}$Lennart Blanken is with Sioux Technologies, Eindhoven, The Netherlands.}
	\thanks{$^{4}$Tom Oomen is with the Delft Center for Systems and Control, Delft University of Technology, Delft, The Netherlands.}}	
	\maketitle
	 \thispagestyle{empty}
	\begin{abstract}
		Fast-sampled models are essential for control design, e.g., to address intersample behavior. The aim of this paper is to develop a non-parametric identification technique for fast-sampled models of systems that have relevant dynamics and actuation above the Nyquist frequency of the sensor, such as vision-in-the-loop systems. The developed method assumes smoothness of the frequency response function, which allows to disentangle aliased components through local models over multiple frequency bands. The method identifies fast-sampled models of slowly-sampled systems accurately in a single identification experiment. Finally, an experimental example demonstrates the effectiveness of the technique.
	\end{abstract}
	
\section{Introduction}
Systems that have actuation and dynamics above the Nyquist frequency of the sensor, known as slow-sampled systems, are becoming increasingly common in for example vision-in-the-loop systems. As a consequence of the Nyquist-Shannon sampling theorem \cite{Shannon1949}, slow-sampled systems are typically identified up to the Nyquist frequency of the slow-sampled sensor. In sharp contrast, fast-sampled models of systems are typically required for control design and performance evaluation, e.g., for the use in evaluating intersample performance \cite{Oomen2007}.
 
Non-parametric frequency-domain representations are often used for performance evaluation and controller design of linear time-invariant (LTI) systems. An example is manual loop-shaping \cite{Schmidt2020} and parametric system identification \cite{Pintelon1994}. A common method for frequency-domain representation is through Frequency Response Functions (FRFs). FRFs can directly be identified from input-output data and are fast, accurate, and inexpensive to obtain \cite{Pintelon2012,Oomen2018a}. Finally, FRFs allow for direct evaluation of stability, performance, and robustness \cite{Skogestad2005}.

The identification of fast-sampled models for slow-sampled systems is challenging, since the maximum achievable identification frequency of traditional FRF identification for LTI systems is limited by the Nyquist frequency of the slow-sampled sensor. The key reason is that fast-sampled outputs are aliased when sampled by a slow-sensor, resulting in indistinguishable contributions in the output, and hence, a fast-sampled model cannot be uniquely recovered \cite{Astrom2011}. As a result, techniques for identifying fast-sampled models for slow-sampled systems, that are required for control design and performance evaluation, are necessary.

Important developments have been made in identification techniques for slow-sampled systems, primarily in continuous-time and multirate parametric system identification. First, continuous-time system identification aims to identify a continuous-time parametric model using input-output data, as outlined in \cite{Unbehauen1990}. Typically, these methods require intersample assumptions on the input signal, e.g., zero-order hold or bandlimited signals \cite{Ljung2009}. Second, parametric identification of slow-sampled systems are developed and include methods for impulse response \cite{Ding2004} and output-error \cite{Zhu2009} model estimation. Lifting techniques, such as using subspace \cite{Li2001}, frequency-domain \cite{VanHaren2022b} or hierarchical identification techniques \cite{Ding2011}, are also developed. These methods focus on parametric identification, require intersample assumptions on the input signal and do not exploit fast-sampled inputs, and consequently, do not disentangle aliased components.

Although methods for identification beyond the Nyquist frequency of slow-sampled systems have been developed, an efficient and systematic methodology for single-experiment FRF identification of fast-sampled models, that disentangles aliased components with arbitrary input signals, is currently lacking. In this paper, slow-sampled systems are identified with excitation signals that cover the full frequency spectrum, where aliased components are disentangled from each other through exploiting the assumption of smooth behavior in the frequency response of a system. Generally, this assumption is at the basis of modern FRF identification, as seen in techniques for LTI single-rate systems, such as Local Polynomial Modeling (LPM) \cite{Pintelon2012} or local rational modeling \cite{McKelvey2012}. In fact, LPM for LTI single-rate systems is recovered as a special case of the developed framework. The key contributions of this paper include the following.
\begin{itemize}
\item[C1] Identification of non-parametric fast-sampled models for slow-sampled systems, by exciting the full frequency spectrum and aliased components are disentangled from each other by assuming smooth behavior in the frequency domain (\secRef{sec:LPM}).
	\item[C2] Validation of the framework for identification of slow-sampled systems in an experimental setup (\secRef{sec:results}).
\end{itemize}

\paragraph*{Notation}
Fast-sampled signals are denoted by subscript $h$ and slow-sampled signals by subscript $l$. The $N$-points and $M$-points Discrete Fourier Transform (DFT) for respectively finite-time fast-sampled and slow-sampled signals are given by
\begin{equation}
	\label{eq:FFT}
	\begin{aligned}
		X_h(k) &= \sum_{\dt=0}^{N-1} x_h(\dt) e^{-j\Omega_k \dt T_{s,h}}, \\
	X_l(k) &= \sum_{m=0}^{M-1} x_l(m) e^{-j \Omega_k m T_{s,l} } \\
&= \sum_{\dt=0}^{M-1} x_h(\dt\fac) e^{-j \Omega_k \dt T_{s,l}},
	\end{aligned}
\end{equation}
with frequency bin $k$, sampling times $\tsh$ and $\tsl$, discrete-time indices for fast-sampled signals $\dt\in \mathbb{Z}_{[0,N-1]}$, slow-sampled signals $m\in \mathbb{Z}_{[0,M-1]}$ with integers $\mathbb{Z}$ and $N,M$ the amount of data points of the fast and slow sampled signals, and generalized frequency variable
\begin{equation}
	\label{eq:omegak}
	\Omega_k=\frac{2\pi k}{N \tsh}=\frac{2\pi k}{M \tsl}.
\end{equation}
The sampling times of the slow-sampled and fast-sampled signals relate as $\tsl=\fac\tsh$, with downsampling factor $\fac\in\mathbb{Z}_{>0}$. The complex conjugate of $A$ is denoted as $\overline{A}$ and the complex conjugate transpose as $A^H$. The complement of sets is given by $A\setminus B$. The expected value of a random variable $X$ is given by $\mathbb{E}\left\{X\right\}$.
\section{Problem Formulation}
In this section, the problem that is considered in this paper is presented. First, the identification setting is presented. Finally, the problem addressed in this paper is defined.
\subsection{Identification Setting}
The goal is to identify a fast-sampled non-parametric model $\hat{G}$ with sampling rate $\fsh=\frac{1}{\tsh}$, using the slow-sampled output $y_l$ with sampling rate $\fsl=\frac{1}{\tsl}=\frac{1}{\fac}\fsh$, where the open-loop control structure considered is visualized in \figRef{fig:setting}. 
\begin{figure}[tb]
	\centering
	\vspace{3mm}\includegraphics[height = 2.75cm]{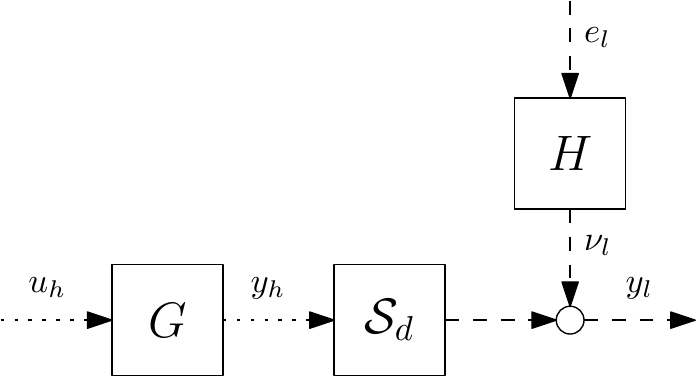}
	\makeatletter
	\caption{Identification setting considered for slow-sampled systems.}
	\label{fig:setting}
\end{figure}
The system $G$ and the noise model $H$ are LTI Single-Input Single-Output (SISO) systems. The input-output behavior of LTI SISO systems in the frequency domain is given by
\begin{equation}
	\label{eq:IO1}
	Y_h(k) = G(\Omega_k)U_h(k) + T_G(\Omega_k),
\end{equation}
with system transient term $T_G(\Omega_k)$ \cite{Pintelon2012}. This reveals that a single frequency of $Y_l$ is influenced by a single frequency of $U_h$, also called the frequency-separation principle. The measured slow-sampled output is a downsampled version of the fast-sampled output as shown in \figRef{fig:setting}, i.e.,
\begin{equation}
		\label{DownsampleFreqDomain}
	Y_l(k) = \mathcal{S}_dY_h(k) + V_l(k),
\end{equation}
with noise $V_l(k)=H(\Omega_k)E(k)$, where $E(k)$ is filtered zero-mean white noise, and is assumed to be independent and identically distributed. In time-domain, the downsampling operation in \eqref{DownsampleFreqDomain} equates to $\mathcal{S}_dy_h(\dt) = y_h(\dt\fac)$. By applying the downsampling operation in the frequency-domain in \eqref{DownsampleFreqDomain}, the DFT of the slow-sampled output is given by \cite{Vaidyanathan1993}
\begin{equation}
	\label{eq:DownsampleFreqDomain}
		Y_l(k) =\frac{1}{F} \sum_{f=0}^{F-1} Y_h\left(k+Mf \right) + V_l(k).
\end{equation}
By substituting the input-output behavior of the fast-sampled system $G(\Omega_k)$ in \eqref{eq:IO1}, the slow-sampled output is given by
\begin{equation}
	\label{eq:DownsampleFreqDomainFinal}
		Y_l(k) \!=\!\! \frac{1}{F} \!\sum_{f=0}^{F-1}\!\! \left(G(\Omega_{k+Mf})U_h(k\!+\!\!Mf) \!+\! T_G(\Omega_{k+Mf})\right)\!+\! V_l(k). 
\end{equation}

\subsection{Problem Definition and Approach}
\label{sec:problemDefApproach}
The DFT $Y_l$ in \eqref{eq:DownsampleFreqDomainFinal}, for a single frequency bin $k$, is influenced by \fac frequencies of $G$ and $U_h$. This is caused by aliasing due to the downsampling operation. Hence, the fast-rate system $G(\Omega_k)$ can in general not be uniquely identified with the slow-sampled output $Y_l$ for arbitrary inputs $U_h$.

The problem considered in this paper is as follows. Given fast-sampled input data $u_h$ and slow-sampled output signal $y_l$ with DFTs $U_h$ and $Y_l$ shown in \eqref{DownsampleFreqDomain}, identify a fast-sampled model of $G(\Omega_k)$ in the frequency-domain for bins $k\in\mathbb{Z}_{[0,N-1]}$, i.e., up until the fast-sampled sampling frequency \fsh, for the identification setup seen in \figRef{fig:setting}. The approach is developed in two steps.
\begin{enumerate}
	\item Development of an intuitive idea in \secRef{sec:approach} for identifying slow-sampled systems using a dedicated input signal for a sparse frequency grid.
	\item Development of the full approach in \secRef{sec:LPM} using arbitrary input signals and full frequency grids, leading to contribution C1.
\end{enumerate}

\section{Intuitive Idea: Identification with Sparse Frequency Spectrum}
\label{sec:approach}
The first step in \secRef{sec:problemDefApproach} is developed, where aliasing is precisely traced for each input signal such that the slow-sampled output $Y_l(k)$ is only influenced by a single fast-sampled input $U_h(k)$. This step is intended for conveying the intuitive idea, leading to the full approach, i.e., the second step in \secRef{sec:problemDefApproach}, in \secRef{sec:LPM}.

The transient contribution $T_G=0$ in this section, meaning that the transient is neglected to facilitate the development of the intuitive idea, in that case \eqref{eq:DownsampleFreqDomainFinal} is equal to
\begin{equation}
	\label{eq:Yl3}
	\begin{aligned}
		Y_l(k) = &\frac{1}{\fac} \sum_{f=0}^{\fac-1} \left(G\left(\Omega_{k+Mf}\right)U_h(k+Mf)\right)+ V_l(k), 
	\end{aligned}
\end{equation}
showing that the slow-sampled output is a summation of the baseband response for $f=0$, and aliased components for $f \in \mathbb{Z}_{[1,\fac-1]}$. Due to the summation of the baseband and aliased contributions, the fast-sampled system can in general not be uniquely recovered from slow-sampled output $Y_l$. An example is seen in the top of \figRef{fig:exampleProblemIdea}.
\begin{figure}[tb]
\centering
\vspace{1mm}\includegraphics{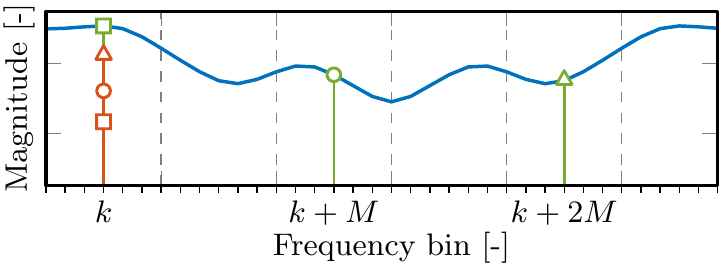}
\includegraphics{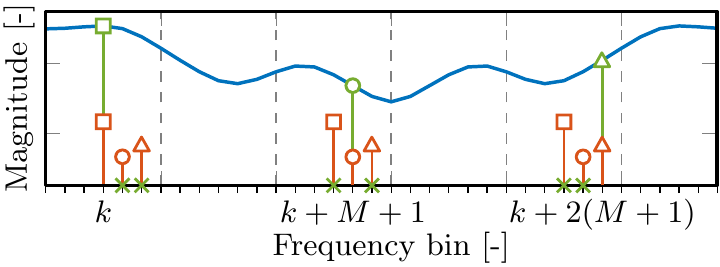}
\caption{Illustration of problem (top) and intuitive idea (bottom) for identifying a fast-sampled model of $G(\Omega_k)$ \li{mblue}{solid} with slow-sampled output $Y_l$. The Nyquist frequency and multiples when sampling the output a factor $\fac=3$ slower is shown as \li{mgray}{densely dashed}[0.6]. Top: Non-zero input $U(k+iM) \:\forall i\in\{0,1,2\}$ and associated gain $|G(\Omega_{k+iM})|$ \li{mgreen}{solid} show that the slow-sampled output $Y_l(k)$ \li{mred}{solid} is a summation as in \eqref{eq:Yl3}, resulting in $\fac=3$ unknowns, but only 1 equation. Bottom: Non-zero input $U(k+i(M+1)) \:\forall i\in\{0,1,2\}$ and associated gain $|G(\Omega_{k+i(M+1)})|$ \li{mgreen}{solid} result in a single contribution of the summation in \eqref{eq:Yl3} influencing $Y_l(k+i)$ \li{mred}{solid} by deliberately not exciting specific bins (\crossMarker{mgreen}). Hence, the fast-sampled system $G(\Omega_{k+i(M+1)})$ can be uniquely recovered at frequency bins $k+i(M+1)$.}
\label{fig:exampleProblemIdea}
\end{figure}

The key idea is that by designing the input signal $U_h$ such that the output at frequency bin $k$ is only influenced by a single contribution of the summation $G\left(\Omega_{k+Mf}\right)U_h(k+Mf)$ in \eqref{eq:Yl3}, the fast-sampled system can be recovered for a subset of all frequency bins. This means that input signals of the form
\begin{equation}
	\label{eq:UConcept}
	\begin{cases}
		U_h(k+Mf) \neq 0,  & f=i, \: i\in\mathbb{Z}_{[0,\fac-1]}, \\
		U_h(k+Mf)=0, &  \mathbb{Z}_{[0,\fac-1]} \setminus i,
	\end{cases}
\end{equation}
result in the input-output behavior
\begin{equation}
	\label{eq:Yl4}
	\begin{aligned}
		Y_l(k) =&\frac{1}{\fac}  G\left(\Omega_{k+iM}\right)U_h(k+iM)+ V_l(k),
	\end{aligned}
\end{equation}
which shows that the summation disappears, and hence, the fast-sampled system $G$ can be uniquely recovered for the frequency bins $k+iM$. 
The sparse set of excited frequencies is from now on denoted by $\mathcal{S}$, i.e., the signal in \eqref{eq:UConcept} can be represented by $U_h(k)\neq 0 \:\forall k\in\mathcal{S}$ and $U_h(k)=0 \: \forall k \notin \mathcal{S}$. The general concept of identifying slow-sampled systems with sparse multisines is shown in the bottom of \figRef{fig:exampleProblemIdea}. By noting that $Y_l(k)=Y_l(k+iM),\: V_l(k)=V(k+iM) \:\forall i\in\mathbb{Z}$, due to the $M$-periodicity of the DFT, \eqref{eq:Yl4} is rewritten as
\begin{equation}
	\label{eq:Yl5}
	\begin{aligned}
		Y_l(k+iM) &=\frac{1}{\fac} G\left(\Omega_{k+iM}\right)U_h(k+iM)+V_l(k+iM), \\
		 &\forall i\in\mathbb{Z}_{[0,\fac-1]},\: (k+iM)\in\mathcal{S}.
	\end{aligned}
\end{equation}
An estimate of the system $G$ is given by
\begin{equation}
	\begin{aligned}
		\hat{G}(\Omega_{k+iM}) &= \fac\frac{Y_l(k+iM)}{U_h(k+iM)},\\
		\forall i\in\mathbb{Z}&_{[0,\fac-1]},\: (k+iM)\in\mathcal{S},
	\end{aligned}
\end{equation}
given that the excitation signal is designed such that the output $Y_l(k)$ is only influenced by a single frequency of $U_h(k)$ according to \eqref{eq:UConcept}. An example of a sparse multisine is shown in \exampleRef{example:sparseMS}.
\begin{example}
	\label{example:sparseMS}	
	An example of a sparse multisine, that achieves broadband excitation, is given by
	\begin{equation}
		\begin{aligned}
			&U_h(k) \neq 0, \quad \forall k\in \mathcal{S}, \\
			&\mathcal{S} = \left\{j+i(M+1)\bigg|  \begin{array}{l}
				i\in\mathbb{Z}_{[0,\fac-1]}, \\
				j\in\left\{0,\fac,2\fac,\ldots,\frac{1}{2}M\right\}
			\end{array}\right\}.
		\end{aligned}
	\end{equation}
\end{example}
Essentially, the sparse multisines in \exampleRef{example:sparseMS} avoid interference of different frequency bands, and by appropriately selecting the inputs a system estimate $\hat{G}$ is obtained at a sparse set of frequencies in each frequency band, including the bands beyond the Nyquist frequency.

The resulting estimate $\hat{G}(\Omega_k)$ is only obtained at the sparse set of excited bins $k\in\mathcal{S}$. By ranging $i$ over multiple sets of experiments in $\{0,1,\ldots\fac-1\}$, the system $\hat{G}(\Omega_k)$ can be uniquely recovered for the full frequency spectrum. As a result, \fac experiments are necessary to identify the system, which leads to a time-intensive procedure. In the next section, a time-efficient single-experiment identification approach is developed.

\section{Identification with Full Excitation Spectrum}
\label{sec:LPM}
In this section, the approach is developed to identify a fast-sampled model of $G(\Omega_k)$ for all frequency bins $k\in\mathbb{Z}_{[0,N-1]}$, given slow-sampled outputs, therewith constituting contribution C1. This is realized by exciting the full frequency spectrum, where aliased contributions are disentangled by exploiting a smoothness condition on $G$ and the transient. In contrast to \secRef{sec:approach}, where for each frequency bin $k$ the \fac unknowns in \eqref{eq:Yl3} were reduced to a single unknown as seen in \eqref{eq:Yl5}, in this section the amount of equations are increased to \fac for each frequency bin $k$. First, the method is presented. Second, a covariance analysis is provided. Finally, the developed approach is summarized in a procedure.
\subsection{Identification of Slow-Sampled Systems with Full Excitation Spectrum}
\label{sec:fullLPM}
The frequency response of the system $G(\Omega_k)$ and transient $T_G(\Omega_k)$ are assumed to be smooth, as is formalized in \assRef{ass:smoothplant} and \assRef{ass:smoothT}. The smoothness assumption enables disentangling aliased components, and consequently, the need for a sparse excitation signal in \secRef{sec:approach} is relaxed.
\begin{assumption}
	\label{ass:smoothplant}
	The frequency response of the fast-sampled system $G(\Omega_k)$ can be approximated in a local window $r\in \vspace{1mm}\mathbb{Z}_{[-n_w,n_w]}$, with $2n_w+1$ the window size, as an $R^{\mathrm{th}}$ order polynomial as
	\begin{equation}
		\label{eq:polymodelG}
		G(\Omega_{k+r}) \approx G(\Omega_{k}) + \sum_{s=1}^{R}g_s(k)r^s,
	\end{equation}
\end{assumption} 
\begin{assumption}
	\label{ass:smoothT}
	The summation of transients, seen in \eqref{eq:DownsampleFreqDomainFinal}, is assumed to be smooth in the local window $r\in\mathbb{Z}_{[-n_w,n_w]}$, i.e.,
	\begin{equation}
		\label{eq:polymodelT}
		\frac{1}{\fac} \sum_{f=0}^{\fac-1}T_G\left(\Omega_{k+r+M f}\right) \approx T(\Omega_k) + \sum_{s=1}^{R} t_s(k)r^s.
	\end{equation}
\end{assumption}
The assumption of a locally smooth system and transient is commonly imposed and at the basis of modern FRF identification, and is valid since $G(\Omega_k)$ and $T_G(\Omega_m)$ are functions with continuous derivatives up to any order \cite{Pintelon2012,McKelvey2012}. Hence, the assumption of a smooth summation of transients is equally reasonable, since the individual transient contributions are smooth. By substituting the polynomial models for $G(\Omega_k)$ in \eqref{eq:polymodelG} and for the transient in \eqref{eq:polymodelT}, the slow-sampled output in \eqref{eq:DownsampleFreqDomainFinal} is rewritten with parameter vector $\Theta$ and data vector $K$ for the local window $r$ as
\begin{equation}
	\label{eq:ykplsur}
	Y_l(k+r) = \Theta(k) K(k+r) + V_l(k+r).
\end{equation}
The parameter vector $\Theta(k)\in\mathbb{C}^{1\times(R+1)(\fac+1)}$ is given by
\begin{equation}
		\Theta(k) = \begin{bmatrix}
			\theta_{G} \!\! & \theta_{g_1} \!\!& \!\!\cdots\!\! & \theta_{g_R} \!\! & T(\Omega_k) \!\! & t_1(k) \!\! & \!\cdots\! & t_R(k)
		\end{bmatrix},
\end{equation}
with
\begin{equation}
	\begin{aligned}
		\label{eq:thetas}
		\theta_{G} &= \frac{1}{\fac}\begin{bmatrix}
			G(\Omega_{k}) &  G(\Omega_{k+M}) &\!\!\!\cdots\!\!\! & G(\Omega_{k+(\fac-1)M})
		\end{bmatrix}, \\
		\theta_{g_i} &= \frac{1}{\fac}\begin{bmatrix}
			g_i(k)  &  g_i(k+M) &\!\!\!\cdots \!\!\! & g_i(k+(\fac-1)M)
		\end{bmatrix},
	\end{aligned}
\end{equation}
and data vector $K(k+r)\in\mathbb{C}^{(R+1)(\fac+1)\times 1}$ is given by
\begin{equation}
	\label{eq:dataVec}
	K(k+r) = \begin{bmatrix}
		K_1(r) \otimes \underline{U}(k+r) \\
		K_1(r)
	\end{bmatrix},
\end{equation}
with input vector 
\begin{equation}
	\label{eq:inputVec}
	\underline{U}(k+r)=\begin{bmatrix}
		U_h(k+r) \\ U_h(k+r+M) \\ \vdots \\ U_h(k+r+(\fac-1)M)\end{bmatrix},
\end{equation}
where $K_1(r)=\begin{bmatrix}1 & r & \cdots & r^R\end{bmatrix}^\top$ and $\otimes$ denotes the Kronecker product. Collecting the column vectors from \eqref{eq:ykplsur} in matrices for the window $r\in\mathbb{Z}_{[-n_w,n_w]}$ gives
\begin{equation}
	\label{eq:Yn}
	Y_{l,n_w} = \Theta(k)K_{n_w} + V_{n_w},
\end{equation}
where $	Y_{l,n_w}\in\mathbb{C}^{1\times2n_w+1}$, $K_{n_w}\in\mathbb{C}^{(R+1)(\fac+1)\times2n_w+1}$ and $V_{n_w}\in\mathbb{C}^{1\times2n_w+1}$ are constructed as
\begin{equation}
	\label{eq:gataherWindow}
		X_{n_w}\!\!\!=\! \begin{bmatrix}
			X(k-{n_w}) \!\!\! & X(k-n_{w}+1)  \!\! & \!\!\!\cdots\!\!\!\! & X(k+n_{w})
		\end{bmatrix}. 
\end{equation}
The fast-sampled system $G(\Omega_k)$ and transient $T(\Omega_k)$ are uniquely identifiable for all $k\in\mathbb{Z}_{[0,N-1]}$, in the presence of aliasing and with a full excitation spectrum, as in \theoremRef{theorem:excLPM}.
\begin{theorem}
	\label{theorem:excLPM}
	Given a frequency bin $k$, if $M+1>2n_w+1\geq(\fac+1)(R+1)$ and the input vector $\underline{U}(k+r)$ from \eqref{eq:inputVec} is designed such that $K_{n_w}$ is of full row rank, an estimate of the parameter vector $\Theta$ in the least-squares sense of \eqref{eq:Yn} is uniquely determined as
	\begin{equation}
		\label{eq:hatTheta}
		\hat{\Theta}(k)=Y_{l,n_w} K_{n_w}^H\left(K_{n_w} K_{n_w}^H\right)^{-1},
	\end{equation}
and the fast-sampled model of system $G(\Omega_k)$ is obtained by
\begin{equation}
	\label{eq:estSystem}
	\hat{G}(\Omega_k) = \fac \hat{\Theta}(k) \begin{bmatrix}
		1 & 0_{1\times((F+1)(R+1)-1)}
	\end{bmatrix}^\top.
\end{equation}
Similarly, $\hat{T}(\Omega_k)=\frac{1}{\fac} \sum_{f=0}^{\fac-1}\hat{T}_G(\Omega_{k+Mf})$ from \eqref{eq:polymodelT} can be obtained.
\end{theorem}
\begin{proof}
	The matrix inverse $\left(K_{n_w} K_{n_w}^H\right)^{-1}$ uniquely exists if $\left(K_{n_w} K_{n_w}^H\right)$ has full rank, which is achieved if the rank of $K_{n_w}$ is equal to the row rank of $K_{n_w}$ and $K_{n_w}$ is full row rank.
\end{proof} 
The interpretation of \theoremRef{theorem:excLPM} is as follows. First, sufficient data $2n_w+1$ should be available, such that \eqref{eq:hatTheta} leads to a unique solution of the $(R+1)(\fac+1)$ parameters, hence $2n_w+1\geq(\fac+1)(R+1)$, which is satisfied by design of wide matrix $K_{n_w}$. In other words, by applying the smoothness assumption, the estimation of the $(R+1)(\fac+1)$ parameters $\hat{\Theta}$ in \eqref{eq:hatTheta} utilizes $2n_w+1$ outputs in $Y_l(k+r)$. This explains how the smoothness assumption allows to disentangle \fac aliased contributions at a frequency bin $k$. Additionally, no overlapping between windows $k+r+iM \:\:\forall i\in\mathbb{Z}_{[0,\fac-1]}$ is allowed, otherwise $K_{n_w}$ is not full row rank, and hence, $M+1>2n_w+1$. Second, the system in \eqref{eq:hatTheta} can be solved uniquely if $K_{n_w}$ is full row rank. As a consequence, aliased and transient contributions can be disentangled if all inputs in the local window and at the aliased windows are sufficiently 'rough', that is formalized for a single local window in \cite{Schoukens2009}. For \theoremRef{theorem:excLPM}, this means that the spectral difference $\left| U(k+r_1+iM)-U(k+r_2+iM)\right|\neq 0 \: \forall r_1,r_2\in\mathbb{Z}_{[-n_w,n_w]},\: \forall i\in\mathbb{Z}_{[0,\fac-1]}$. This condition is fulfilled by, e.g., random-phase multisines \cite{Pintelon2012}. The developed framework is illustrated in \figRef{fig:LPMconcept}.
\begin{remark}
	\label{recoverLPM}
	Note that traditional LPM for single-rate LTI systems is recovered as a special case of the developed framework by setting $\fac=1$.
\end{remark}
\begin{figure}[tb]
	\centering\vspace{2mm}\includegraphics{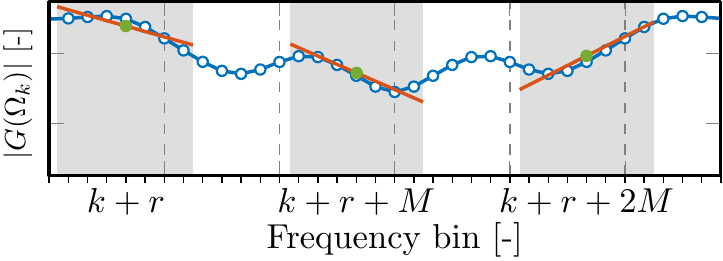}
	\caption{Illustration of identification of slow-sampled system that disentangles aliased components by assuming local smoothness. True fast-sampled system $G(\Omega_k)$ (\circleLine{mblue}), and the local first-order parametric estimates $\hat{{G}}(\Omega_{k+r+iM})=\hat{{G}}(\Omega_{k+iM})+g_1(k+iM)r \:\forall i\in\{0,1,2\}$ of the developed approach \li{mred}{solid}$\!$. The values $\fac=3$, $R=1$ and $n_w=3$, when neglecting the transient, result in 6 unknowns, i.e., the system estimates $\hat{{G}}(\Omega_{k+iM})$ (\filledCircle{mgreen}[2]) and the polynomial coefficients $g_1(k+iM)$, that determine the slopes of \li{mred}{solid}$\!$. Since $n_w=3$, there are in total $2n_w+1=7$ equations, hence the system of equations in \eqref{eq:hatTheta} can be solved for the 6 unknowns.} 
	\label{fig:LPMconcept}
\end{figure}
\begin{remark}
	The identified FRF by sparse multisines from \secRef{sec:approach} can be interpolated at the non-excited frequency bins $k\notin \mathcal{S}$, similar to \cite{Geerardyn2017}, by seeing them as a special case of the framework in this section. The condition on $\underline{U}(k+r)$ in \theoremRef{theorem:excLPM} is inherently satisfied by sparse excitation.
\end{remark}
\subsection{Variance Estimate of the FRF}
The developed framework enables estimation of the variance of the identified FRF. The variance of ${\hat{G}}$, that is estimated using \eqref{eq:hatTheta}, is given by \theoremRef{theorem:var}. \\[0 mm]
\begin{theorem}
	\label{theorem:var}
	The estimated variance of the FRF $\hat{G}$ that is estimated using \eqref{eq:estSystem} is given by
	\begin{equation}
		\label{eq:estCov}
		\varn\left(\hat{{G}}\left(\Omega_k\right) \right)  \approx \fac \overline{ S^H {S}} \hat{C}_v(k),
	\end{equation}
	that is an estimate of the true variance of the identified FRF
	\begin{equation}
		\label{eq:trueCov}
		\varn\left(\hat{{G}}\left(\Omega_k\right) \right) = \fac \mathbb{E}\left\{ \overline{ S^H {S}}\right\}C_V(k) +\fac O_{int\: H}\left( \frac{n_w^0}{M}\right), 
	\end{equation}
with $C_V$ the variance of the noise and an estimate based on measurements $\hat{C}_V$, noise interpolation error $O_{int\: H}$ \cite{Pintelon2012}, and
\begin{equation}
	S = K_{n_w}^H\left(K_{n_w}K_{n_w}^H \right) ^{-1}\begin{bmatrix}
		1 & 0
	\end{bmatrix}^\top.
\end{equation}
\end{theorem}
\begin{proof}
	The proof extends \cite[Appendix~7.E]{Pintelon2012} to \fac frequency bands. In particular, by combining \eqref{eq:Yn} and \eqref{eq:hatTheta} into
		\begin{equation}
			\label{eq:residualProof}
		\begin{aligned}
			\left( Y_{n_w} - \Theta K_{n_w}\right)S = V_{n_w}S&,  \\
			\underbrace{\hat{\Theta}\begin{bmatrix}
					1 &
					0
			\end{bmatrix}^\top}_{\frac{1}{\fac}\hat{{G}}} - \underbrace{\Theta\begin{bmatrix}
					1 &
					0
			\end{bmatrix}^\top}_{\frac{1}{\fac}{G}}  = V_{n_w}S&, \\
			\hat{{G}} = {G} + \fac V_{n_w}S&,
		\end{aligned}
	\end{equation}
the factor \fac appears in the difference between the true and estimated system $G$ and $\hat{{G}}$. 
\end{proof}

\noindent The variance of the noise is equal to
\begin{equation}
	C_V(k) = \varn\left(V_l(k) \right) =	\mathbb{E}\left\{V_l(k) V_l^H(k)\right\}.
\end{equation}
An estimate of the noise variance is calculated by taking an average over the local window, see \cite[Appendix~7.B]{Pintelon2012} for technical details, i.e.,
\begin{equation}
	\label{eq:estCv}
	\hat{C}_V(k) =\frac{1}{2n_w+1-(R+1)(\fac+1)} \hat{V}_{n_w} \hat{V}_{n_w}^H,
\end{equation}
with $\hat{V}_{n_w}=Y_{l,n_w}-\hat{\Theta}(k)K_{n_w}$.
\subsection{Developed Procedure}
The developed approach is summarized in Procedure~\ref{proc:1}, that links the main results in this paper.
\begin{figure}
\vspace{6pt}\hrule \vspace{1mm}\begin{proced}[Identifying slow-sampled systems with full frequency spectrum] \hfill \vspace{0.5mm} \hrule \vspace{1mm}
		\label{proc:1}
		\begin{enumerate}
			\item Construct $u_h$ such that it satisfies the requirements of $\underline{U}$ given in \theoremRef{theorem:excLPM}.
			\item Apply input $u_h$ to system and record the output $y_l$.
			\item Take the DFT of input $u_h$ and output $y_l$ using \eqref{eq:FFT}.
			\item For all frequency bins $k\in\mathbb{Z}_{[0,N-1]}$, do the following.
			\begin{enumerate}
				\setlength\itemsep{0.1em}
				\item Construct matrices $K_{n_w}$ and $Y_{l,n_w}$ from \eqref{eq:Yn} using \eqref{eq:dataVec} and \eqref{eq:gataherWindow}.
				\item Compute parameter vector $\hat{\Theta}(k)$ from \eqref{eq:hatTheta}.
				\item Calculate the estimated FRF $\hat{G}(\Omega_k)$ and variance $\varn \left(\hat{G}(\Omega_k)\right)$ using \eqref{eq:estSystem} and \eqref{eq:estCov}.
			\end{enumerate}
		\end{enumerate}
		\vspace{0pt} 	\hrule \vspace{-20pt}
	\end{proced}
\end{figure}

\section{Validation}
\label{sec:results}
In this section, the developed method is validated on an experimental setup, leading to contribution C2.
\subsection{Measurement Setup}
The experimental setup is shown in \figRef{fig:setup}. The setup consists of two rotating masses connected via a rubber band. The rotating masses are each actuated by a DC motor. The input $u_h$ and output $y_l$ are respectively the torque and rotation of the first mass. The second mass is virtually suspended to the fixed world by a feedback controller. The excitation signal for the developed approach that excites the full frequency spectrum is a random-phase multisine and has an root-mean-square value of $8.3\cdot10^{-3}$ Nm. For comparison purposes, the intuitive idea from \secRef{sec:approach} uses sparse multisines with a root-mean-square value of $9.6\cdot10^{-3}$ Nm and are designed as in \exampleRef{example:sparseMS}. A photograph and a schematic overview of the test setup are seen in \figRef{fig:setup}. The settings used during identification are shown in \tabRef{tab:settings}.
\begin{figure}[tb]
	\centering
	\setlength{\fboxsep}{0pt}%
	\setlength{\fboxrule}{0.75pt}%
	\vspace{2mm}\fbox{\includegraphics[height = 1.5cm]{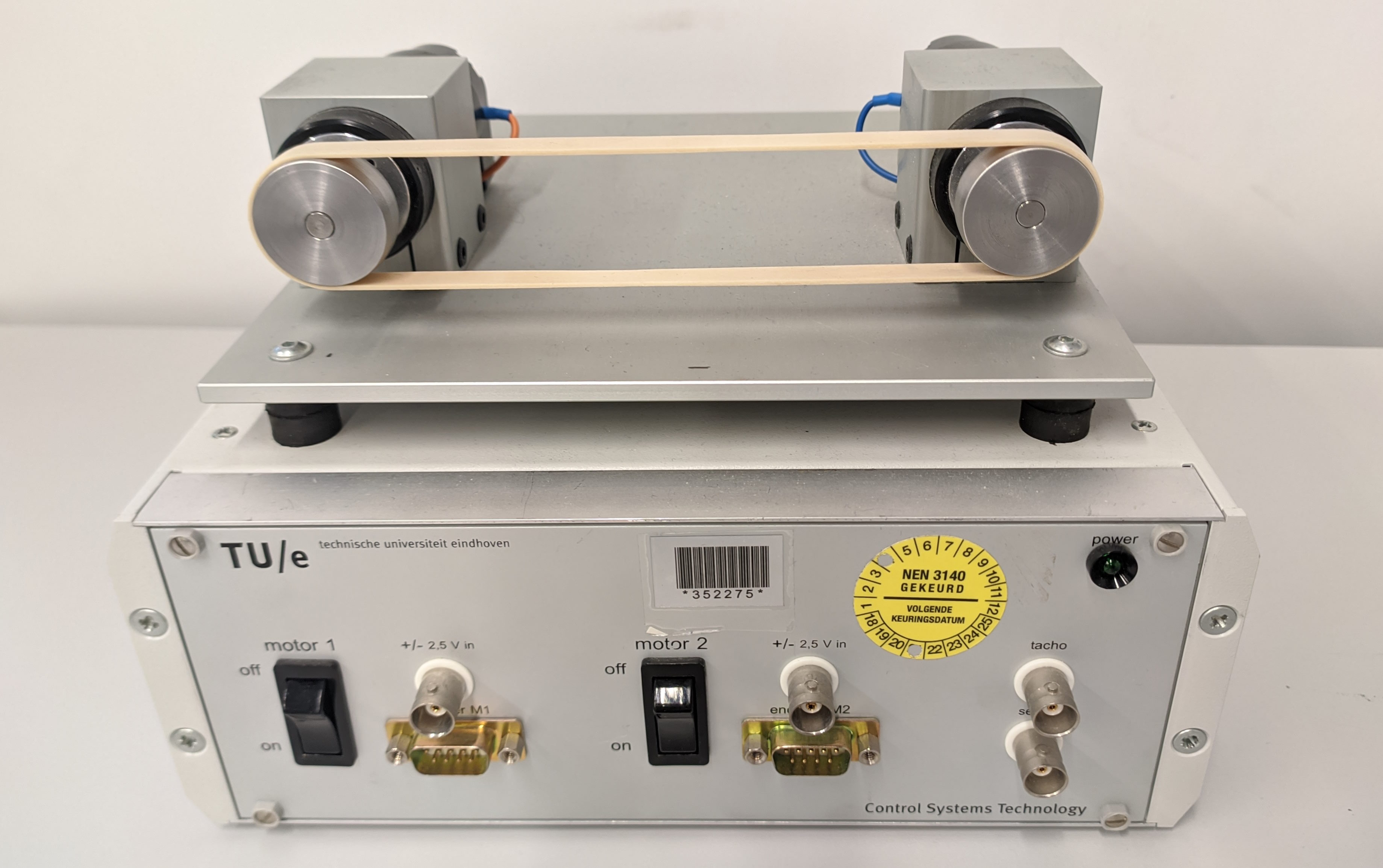}}\hspace{1mm}
	\includegraphics[height = 1.5cm]{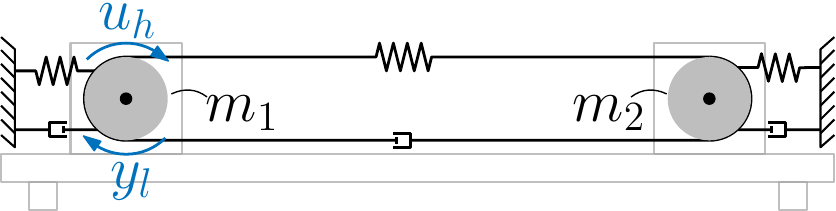} %
	\caption{Left: Picture of experimental setup used. Right: Schematic overview of experimental setup.}
	\label{fig:setup}
\end{figure}
\begin{table}[h] \centering
	\caption{Experimental settings.}
	\label{tab:settings}
	\begin{tabular}{@{}lll@{}}\toprule
		\textbf{Property}    & \textbf{Variable} & \textbf{Value} \\
		\midrule
		Fast sampling rate   & $f_{s,h}$         & 120 Hz         \\
		Slow sampling rate    & $f_{s,l}$         & 30 Hz          \\
		Downsampling factor  & $F$               & 4              \\
		Measurement time     & $T_m$               & 120 s           \\
		Window size          & $n_w$             & 150             \\
		Polynomial degree    & $R$               & 2     		\\        
		\bottomrule
	\end{tabular}
	\vspace{-2mm}
\end{table}
\subsection{Experimental Results}
The fast-sampled system is identified using the developed approach from \secRef{sec:LPM}. For comparison purposes, the sparse multisine approach from \secRef{sec:approach}, a traditional approach and using the fast-sampled output $y$, that is not available for the other approaches, are used to identify the fast-sampled system. The results for the developed, sparse multisine and traditional approach are seen in \figRef{fig:expresultsLPM}, \figRef{fig:expresultsSparse} and \figRef{fig:expresultsTraditional}.\begin{figure}[t]
	\centering
	\includegraphics{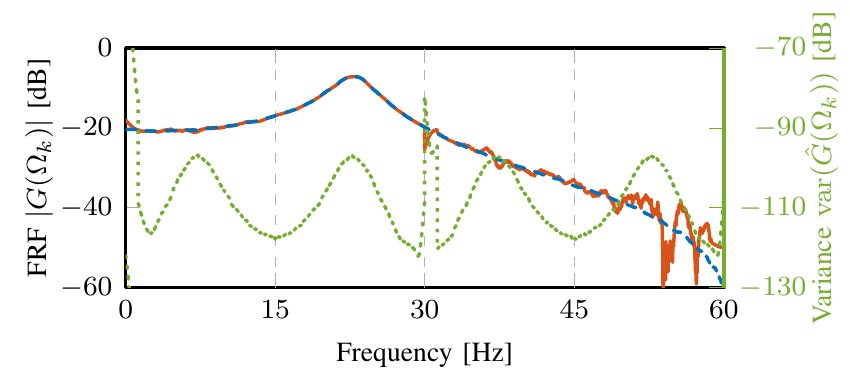} \vspace{-7.8mm}
	\caption{Identified FRF $\hat{G}(\Omega_k)$ for excitation by random-phase multisines covering the full frequency spectrum from \secRef{sec:LPM} \li{mred}{solid} with covariance estimate from \eqref{eq:estCov} \li{mgreen}{dotted} (right axis). The identified FRF based on fast-sampled data is shown as \li{mblue}{densely dashed} and multiples of the Nyquist frequency of the slow sensor as \li{mgray}{dashed}$\!$.} \vspace{-5mm}
	\label{fig:expresultsLPM}
\end{figure}
\begin{figure}[t]
	\centering
	\hspace{-11mm}
	\includegraphics{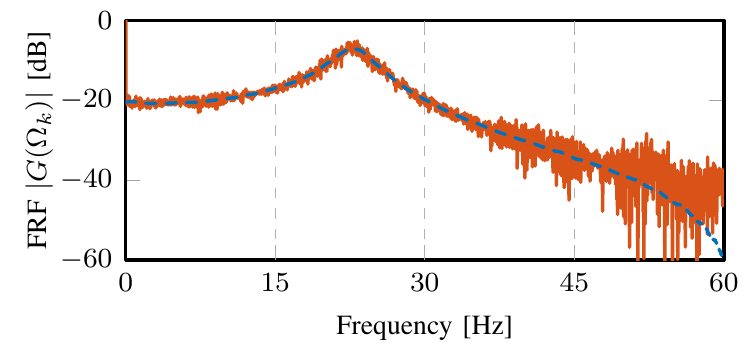} \vspace{-4mm}
	\caption{Identified FRF $\hat{G}(\Omega_k)$ with sparse multisines from \secRef{sec:approach} in a single identification experiment \li{mred}{solid}$\!\!$. The identified FRF based on fast-sampled data is shown as \li{mblue}{densely dashed} and multiples of the Nyquist frequency of the slow sensor as \li{mgray}{dashed}$\!$.} \vspace{-5mm}
	\label{fig:expresultsSparse}
\end{figure}
\begin{figure}[t]
	\centering
	\vspace{2.8mm}\includegraphics{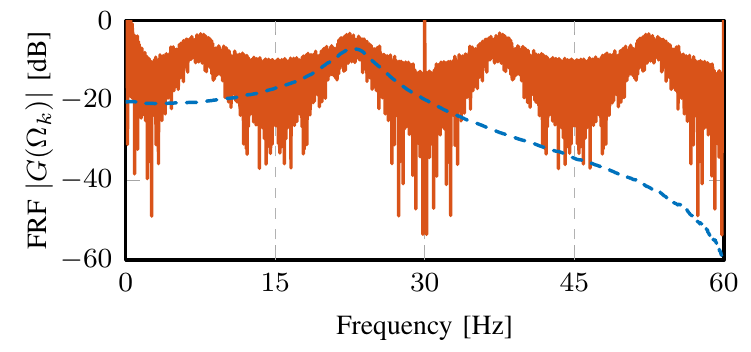} \vspace{-4mm}
	\caption{Identified FRF $\hat{G}(\Omega_k)$ for excitation by full frequency spectrum random-phase multisines by performing $\hat{G}(\Omega_k)={Y_l(k)}/{U_h(k)}$ using the same dataset as the full approach \li{mred}{solid}$\!$. The identified FRF based on fast-sampled data is shown as \li{mblue}{densely dashed} and multiples of the Nyquist frequency of the slow sensor as \li{mgray}{dashed}$\!$.} \vspace{-5mm}
	\label{fig:expresultsTraditional}
\end{figure}
The following observations are made
\begin{enumerate}
	\item From \figRef{fig:expresultsLPM} and \figRef{fig:expresultsSparse}, it is observed that the developed approach and the sparse multisines are capable of identifying the dynamics above the Nyquist frequency of the sensor. The developed approach of \secRef{sec:LPM} that assumes smooth behavior in the frequency-domain has significantly lower variance and a factor \fac higher frequency resolution.
	\item From \figRef{fig:expresultsTraditional}, it is observed that exciting the system with the full frequency spectrum and performing FRF identification by $\hat{{G}}(\Omega_k)={Y_l(k)}/{U_h(k)}$ cannot accurately identify the fast-sampled system $G$, due to aliasing. Additionally, the estimated FRF is dominated by a periodic behavior, which is explained because the slow-sampled DFT $Y_l(k)$, that is $M$-periodic, is used for all frequencies and aliasing is not accounted for.
	\item From \figRef{fig:expresultsLPM}, the variance of the developed approach that excites the full frequency spectrum repeats every \fsl. For example, the increase of variance around 23 Hz repeats at 53 Hz. This is explained because the variance from \eqref{eq:estCov} is determined with the estimated noise variance $\hat{C}_v$ from \eqref{eq:estCv}, that uses the slow-sampled data $Y_l$, and hence, is periodic in \fsl. Additionally, the mirroring effect that is observed, e.g., the increase of variance at 23 Hz is mirrored to 7 and 37 Hz, is caused because the DFT of $\hat{C}_v$ is symmetrical in $\frac{1}{2}\fsl$.
\end{enumerate}

\vspace{-2mm}
\section{Conclusions}
The results in this paper enable identifying FRFs of slow-sampled systems where aliasing occurs. The key step is assuming smooth behavior of the system FRF, which allows to appropriately disentangle aliased contributions when exciting the full frequency spectrum. Furthermore, covariance estimates of the FRF are provided. Finally, the framework is validated through experimental results. The dual case, where outputs are fast-sampled and inputs slow-sampled, is trivial by assuming appropriate interpolator behavior. The developed approach is a key enabler for closed-loop, multivariable and parametric system identification and control design for slow-sampled systems, such as vision-in-the-loop systems.

\vspace{-4mm}
\bibliographystyle{IEEEtran}
\bibliography{../../library,BSTControl}


\end{document}